\documentclass[10pt,journal]{IEEEtran}

\usepackage{enumitem}
\usepackage{amsmath,amssymb,amsthm,mathtools}
\usepackage{graphicx}
\usepackage{cite}
\usepackage[colorlinks=true,linkcolor=blue,citecolor=blue,urlcolor=blue]{hyperref}

\newtheorem{theorem}{Theorem}
\newtheorem{proposition}[theorem]{Proposition}

\theoremstyle{remark}

\newcommand{\bbC}{\mathbb{C}}              
\newcommand{\E}{\mathbb{E}}                

\newcommand{\bH}{\mathbf{H}}               
\newcommand{\bx}{\mathbf{x}}               
\newcommand{\by}{\mathbf{y}}               
\newcommand{\bv}{\mathbf{v}}               
\newcommand{\bA}{\mathbf{A}}               
\newcommand{\bM}{\mathbf{M}}
\newcommand{\bN}{\mathbf{N}}
\newcommand{\bI}{\mathbf{I}}               
\newcommand{\bV}{\mathbf{V}}               
\newcommand{\bP}{\mathbf{P}}               
\newcommand{\bPhi}{\mathbf{\Phi}}          
\newcommand{\bLambda}{\mathbf{\Lambda}}    

\newcommand{\Nt}{N_{\mathrm{t}}}           
\newcommand{\Nr}{N_{\mathrm{r}}}           

\newcommand{\SNR}{\mathsf{SNR}}

\newcommand{\C}{\mathsf{C}}                    

\newcommand{\cV}{\mathcal{V}}              
\newcommand{\cI}{\mathcal{I}}

\newcommand{\lamcut}{\lambda_{\mathrm{cut}}}  
\newcommand{\Cl}{\operatorname{Cl}_2}         
\newcommand{\Li}{\operatorname{Li}_2}         
\newcommand{\jj}{\mathrm{j}}                  
\newcommand{\dd}{\,\mathrm{d}}                
\DeclareMathOperator{\tr}{tr}
\DeclareMathOperator{\diag}{diag}

\providecommand{\figref}[1]{Fig.~\ref{#1}}
\providecommand{\secref}[1]{Section~\ref{#1}}

\providecommand{\proprefs}[2]{Propositions~\ref{#1}--\ref{#2}}
\providecommand{\proprefand}[2]{Propositions~\ref{#1} and~\ref{#2}}

\begin{document}

\title{A Finite-SNR Closed Form for the Full-CSI Capacity
of the Square Mar\v{c}enko--Pastur MIMO Channel}

\author{Mohamed~Akrout,~\IEEEmembership{Member,~IEEE}, and~Robert~W.~Heath,~\IEEEmembership{Fellow,~IEEE}
\thanks{M. Akrout is with the Department of Electrical Engineering and Computer Science, University of Tennessee, Knoxville, TN 37996, USA. R. W. Heath, Jr. is with the Department of Electrical and Computer Engineering, University of California, San Diego, La Jolla, CA 92093, USA.
(e-mails: makrout@tennessee.edu; rwheathjr@ucsd.edu).}}

\markboth{}%
{Akrout \MakeLowercase{\textit{et al.}}: Finite-SNR Closed Form for the Square Mar\v{c}enko--Pastur MIMO Capacity}

\maketitle

\begin{abstract}
For the canonical independent and identically distributed (IID) single-user multiple-input multiple-output (MIMO) channel with full channel state information at the transmitter (CSIT), the asymptotic capacity per receive antenna is a waterfilling integral over the Mar\v{c}enko--Pastur law, whose shape depends only on the ratio of transmit to receive antennas. Closed forms are known for every ratio other than the square one, and there only above a finite signal-to-noise ratio (SNR). The square ratio behaves differently because the Mar\v{c}enko--Pastur support reaches the origin. The waterfilling cutoff then stays strictly inside the support at every finite SNR, and the water level has been available only numerically. This letter closes that gap with a classical trigonometric parametrization of the square law, which removes the edge singularity and replaces the moving cutoff by a single angle. Both the SNR and the capacity become explicit functions of that angle, so the capacity curve is traced by sweeping the angle rather than by solving a scalar constraint at each operating point. The same parametrization delivers the low-SNR behavior in closed form, where the slope of capacity in SNR equals the upper edge of the limiting spectrum.

\end{abstract}


\IEEEpeerreviewmaketitle

\section{Introduction}\label{sec:intro}
\subsection{Motivation and prior work}
The capacity of a multiple-input multiple-output (MIMO) link is governed by the limiting spectrum of the channel Gram matrix \cite{tse2005fundamentals}. For the canonical independent and identically distributed (IID) channel between $\Nt$ transmit and $\Nr$ receive antennas, that spectrum is the Mar\v{c}enko--Pastur law, whose shape is set by the aspect ratio $\Nt/\Nr$. The full-CSIT capacity is the waterfilling functional over that law \cite{TulinoVerduRandomMatrixTheoryWireless2004}. When $\Nt\neq \Nr$, the functional admits compact closed forms once the signal-to-noise ratio (SNR) exceeds a finite threshold \cite{GrantSomeResultsMultiAntenna2002, TulinoLozanoVerduMIMOCapacityChannel2004}. The square case $\Nt = \Nr$ behaves differently. Its Mar\v{c}enko--Pastur support reaches the origin, so the water level never floods the entire spectrum at any finite SNR. The capacity has therefore been available in closed form only in the infinite-SNR limit \cite{TulinoVerduRandomMatrixTheoryWireless2004}. Prior work at finite SNR, however, characterizes the capacity through a waterfilling integral whose water level must be found numerically \cite{GrantSomeResultsMultiAntenna2002, JayaweeraPoorCapacityMultipleAntenna2003}.

\subsection{Contributions}
In this letter, we derive the finite-SNR full-CSIT capacity of the square canonical MIMO channel in closed form. The contributions are summarized as follows:
\begin{itemize}
\item We identify what prevents the square ratio from admitting the closed-form treatments that resolve every other ratio. Away from the square ratio, the spectral density is bounded away from zero, so raising the SNR eventually floods the whole support and removes the waterfilling cutoff from the problem. At the square ratio, the density reaches the origin. The cutoff then stays an interior point of the support at every finite SNR, so the capacity integral never reduces to a full-support one.
\item We derive closed-form parametric expressions for the square ratio, based on a classical trigonometric parametrization of the Mar\v{c}enko--Pastur law in which the SNR and the capacity are both explicit functions of a single angle.
\item We verify the closed form against numerical waterfilling over the full SNR range and against the exact high-SNR asymptote.
\end{itemize}

\subsection{Outline}
\secref{sec:model} sets up the model and the waterfilling functional. \secref{sec:cases} explains why the non-square ratios are solvable and the square one is not. \secref{sec:angular} derives the closed-form expressions, \secref{sec:verify} verifies it numerically, and \secref{sec:remarks} concludes.

\section{Background on full-CSI MIMO capacity}\label{sec:model}
The full-CSIT capacity of a MIMO link reduces to a waterfilling integral against the
limiting spectrum of the channel Gram matrix. Carrying that reduction out carefully is what
makes the obstacle at the square ratio visible. In this section, we state the signal model, the
Shannon transform, and the waterfilling solution in the form used throughout the letter.

\subsection{System model}
Consider a single-user link with $\Nt$ transmit and $\Nr$ receive antennas. The received signal is
\begin{align}
\label{eq:channel}
\by=\bH\bx+\bv,
\end{align}
where $\bH\in\bbC^{\Nr\times \Nt}$ collects the fading coefficients between each transmit and each
receive antenna, $\bx\in\bbC^{\Nt}$ is the transmitted vector, and $\bv\in\bbC^{\Nr}$ is additive
noise. We study the large-antenna regime in which $\Nt$ and $\Nr$ grow without bound at a fixed
ratio, the \emph{aspect ratio}
\begin{align}
\label{eq:beta}
\beta=\lim_{\Nt,\Nr\to\infty}\,\frac{\Nt}{\Nr}.
\end{align}

\noindent Without loss of generality, channel gains are scaled so that $\E\!\left[\tr\{\bH\bH^*\}\right]=\Nr$ and
have the variance of each entry equal to $1/\Nt$ when the entries are identically distributed. The average received signal-to-noise ratio ($\SNR$) per observation is then
\begin{align}
\label{eq:snrdef}
\SNR=\frac{\E[\|\bx\|^2]}{\tfrac{1}{\Nr}\,\E[\|\bv\|^2]} .
\end{align}
Because the antennas may transmit correlated streams, the second-order statistics of the input
matter. Normalized by its energy per dimension, the input covariance is
\begin{align}
\label{eq:Phidef}
\bPhi=\frac{\E[\bx\bx^*]}{\tfrac{1}{\Nt}\,\E[\|\bx\|^2]},
\end{align}
so that $\E[\tr\{\bPhi\}]=\Nt$. Diagonalizing $\bPhi=\bV\bP\bV^*$ separates the transmission into $1)$ \emph{signaling
directions} (the columns of the unitary matrix $\bV$) and $2)$ a \emph{power allocation} represented by the
nonnegative diagonal matrix $\bP$ whose $j$th entry is the power assigned to the $j$th direction \cite{HeathLozanoFoundationsMIMOCommunication2018}.

\subsection{Channel-state information and capacity functional}
The capacity-achieving input covariance $\bPhi$, and hence the capacity itself, depends on how much the transmitter
knows about $\bH$ \cite{GoldsmithJafarJindalVishwanathCapacityLimits2003}. Three regimes are typically considered:
\begin{itemize}
\item \textit{Full CSIT:} the transmitter knows the realization of $\bH$ instantaneously, so
$\bPhi$ may be adapted to it. This is the relevant regime for fixed and low-mobility links and for
reciprocal (time-duplexed) systems.
\item \textit{Statistical CSIT:} the transmitter knows only the distribution of $\bH$; then $\bV$
aligns with the eigenvectors of $\E[\bH^*\bH]$ and $\bP$ is found by an iterative
waterfilling.
\item \textit{No CSIT:} the choice is the isotropic input $\bPhi=\bI$.
\end{itemize}
Closed-form capacity expressions at $\beta=1$ are available for the isotropic input \cite{VerduShamaiSpectralEfficiencyCDMARandom1999, RapajicPopescuInformationCapacityRandom2000}. The open case addressed here is the waterfilling input with full CSIT, for which only the $\SNR\to\infty$ limit was known.

To define the capacity in a compact manner, we use the \emph{Shannon transform}. Let $\bA$ be an $n\times n$
Hermitian nonnegative-definite random matrix with (real, nonnegative) eigenvalues
$\lambda_1(\bA),\dots,\lambda_n(\bA)$. As $n\to\infty$, the empirical distribution of these eigenvalues is assumed to converge to a deterministic limiting law \cite{TulinoVerduRandomMatrixTheoryWireless2004}. Let $X\ge0$ denote a random variable distributed
according to that limiting law to represent an eigenvalue of $\bA$ drawn from its limiting spectrum. The Shannon transform $\cV_{\bA}(\SNR)$ is the asymptotic
per-dimension mutual information of a Gaussian channel whose eigenmodes have gains equal to the
eigenvalues of $\bA$ \cite{TulinoVerduRandomMatrixTheoryWireless2004}:
\begin{align}
\cV_{\bA}(\SNR)&=\lim_{n\to\infty}\frac1n\,\E\big[\log\det(\bI+\SNR\,\bA)\big]\notag\\
&=\E\big[\log(1+\SNR\,X)\big].\label{eq:shannonT}
\end{align}
In \eqref{eq:shannonT}, the right-hand side is the average of the scalar map $x\mapsto\log(1+\SNR\,x)$ over the limiting
spectrum of $\bA$. Here, $x$ denotes a value taken by the random eigenvalue $X$, and the
expectation is over that limiting law. In other words, one can rewrite \eqref{eq:shannonT} as
\begin{align}
    \cV_{\bA}(\SNR)=\int\log(1+\SNR\,x)\,
\dd F_{\bA}(x)
\end{align}
where $F_{\bA}$ is the limiting eigenvalue distribution. The per-receive-antenna capacity is the largest Shannon transform of the effective
channel Gram matrix $\bH\bPhi\bH^*$ over admissible inputs,
\begin{align}
\label{eq:shannon}
\C(\SNR)=\max_{\bPhi:\,\tr\{\bPhi\}=\Nt}\ \cV_{\bH\bPhi\bH^*}(\SNR).
\end{align}
Equation~\eqref{eq:shannon} is a \emph{functional} of the channel in the variational sense: its
argument is the matrix-valued input covariance $\bPhi$. For each admissible $\bPhi$, the Shannon transform
$\cV_{\bH\bPhi\bH^*}$ returns the rate that input achieves, and the capacity is the
supremum of that rate subject to the power constraint $\tr\{\bPhi\}=\Nt$. We now specialize \eqref{eq:shannon} to the full-CSIT regime, which is the
subject of this work.

\subsection{Full-CSIT waterfilling solution}\label{sec:wf}
With full CSIT, the maximizer of \eqref{eq:shannon} is classical \cite{HeathLozanoFoundationsMIMOCommunication2018, TelatarCapacityMultiAntennaGaussian1999}: $\bV$ diagonalizes
$\bH^*\bH$, and $\bP$ is the waterfilling allocation over its eigenvalues
$\lambda_1,\dots,\lambda_{\Nt}$,
\begin{align}
\label{eq:wf}
[\bP]_{j,j}=\Big(\nu-\frac{1}{\SNR\,\lambda_j}\Big)^{\!+},
\qquad \tr\{\bP\}=\Nt .
\end{align}
Here $(\cdot)^+=\max\{\cdot,0\}$ and $\nu$ is the \emph{water level} which is a single scalar that is common to
all directions and chosen so that the total-power constraint $\tr\{\bP\}=\Nt$ holds. A direction is
allocated power only if its channel gain is strong enough, i.e., $\lambda_j>1/(\SNR\,\nu)$.
For the full CSIT case, the input covariance optimizer is $\bPhi=\bV\bP\bV^*$ \cite{HeathLozanoFoundationsMIMOCommunication2018}. After writing the eigendecomposition $\bH^*\bH=\bV\bLambda\bV^*$ with
$\bLambda=\diag(\lambda_1,\dots,\lambda_{\Nt})$ and using the identity
$\det(\bI+\bM\bN)=\det(\bI+\bN\bM)$, \eqref{eq:shannon} becomes
\begin{align}
    \C(\SNR)&= \frac{1}{\Nr}\log\det\!\big(\bI+\SNR\,\bH\bPhi\bH^*\big)\notag\\
    &=\frac{1}{\Nr}\log\det\!\big(\bI+\SNR\,\bP\bLambda\big)\notag\\
    &=\frac{1}{\Nr}\sum_{j=1}^{\Nt}\log\big(1+\SNR\,[\bP]_{j,j}\,\lambda_j\big)\notag\\
    &=\underbrace{\frac{\Nt}{\Nr}}_{\to\,\beta}\cdot
\underbrace{\frac{1}{\Nt}\sum_{j=1}^{\Nt}\big(\log(\SNR\,\nu\,\lambda_j)\big)^+}
_{\text{average over the spectrum of }\bH^*\bH}.\label{eq:shannon2}
\end{align}
The inner factor in \eqref{eq:shannon2} is an average of a fixed function over the eigenvalues
$\lambda_1,\dots,\lambda_{\Nt}$ of $\bH^*\bH$. Encoding those eigenvalues in the empirical spectral distribution
\begin{align}
\label{eq:esd}
F^{(\Nt)}_{\bH^*\bH}(x)=\frac{1}{\Nt}\,\#\{\,j:\lambda_j\le x\,\},
\end{align}
i.e.\ the fraction of eigenvalues not exceeding $x$, the inner average is exactly
$\int\big(\log(\SNR\,\nu\,\lambda)\big)^+\dd F^{(\Nt)}_{\bH^*\bH}(\lambda)$. As
$\Nt,\Nr\to\infty$ with $\Nt/\Nr\to\beta$, the empirical distribution \eqref{eq:esd} converges
almost surely to the Mar\v{c}enko--Pastur law with density $f_\beta$ \cite{TulinoVerduRandomMatrixTheoryWireless2004}. Taking this limit yields the almost-sure per-antenna capacity
\begin{align}
\label{eq:cap-limit}
\C(\SNR)=\beta\int\big(\log(\SNR\,\nu\,\lambda)\big)^+ f_\beta(\lambda)\dd \lambda.
\end{align}
Equation \eqref{eq:cap-limit} is the starting point for \secref{sec:cases}. The whole
difficulty of the $\beta=1$ case sits in the behavior of $f_\beta$ near $\lambda=0$ inside this
integral.

\section{Mar\v{c}enko--Pastur law with three aspect-ratio cases}\label{sec:cases}
Whether the waterfilling integral has a closed form is decided by a single feature of the
Mar\v{c}enko--Pastur law, namely how far its support sits from the origin. This section makes
that dependence explicit and isolates the square ratio as the one case where the distance
vanishes.

\subsection{Canonical channel and its limiting spectrum}
Consider the canonical channel $\bH$ having IID zero-mean complex Gaussian entries. The empirical spectral distribution \eqref{eq:esd} of $\bH^*\bH$ converges
almost surely, as $\Nt,\Nr\to\infty$ with $\Nt/\Nr\to\beta$, to the Mar\v{c}enko--Pastur law of ratio $\beta$ \cite{TulinoVerduRandomMatrixTheoryWireless2004}. This is the limiting spectrum of large IID Gram matrices: it is supported on a single interval $[a,b]$ and has density
\begin{align}
\label{eq:mpedges}
f_\beta(\lambda)=\frac{\sqrt{(\lambda-a)(b-\lambda)}}{2\pi\beta\lambda},
~ \text{where } 
\begin{cases}
a=(1-\sqrt\beta)^2,\\
b=(1+\sqrt\beta)^2.
\end{cases}
\end{align}

Two features of \eqref{eq:mpedges} decide everything that follows. First, both support edges depend on $\beta$, and the lower edge $a$ is the one that matters, since it measures how far the spectrum sits from the origin. Second, the density has an explicit $1/\lambda$ factor, which is of no consequence as long as the support stays away from the origin ($a>0$), but becomes decisive when $a\to0$.

\subsection{Full-CSIT capacity and waterfilling cutoff}
Specializing \eqref{eq:cap-limit} to the Mar\v{c}enko--Pastur law, and writing the positive part as a lower
integration limit, the asymptotic full-CSIT capacity is
\begin{align}
\label{eq:cap-general}
\C(\SNR)=\beta\int_{\max\{a,\lamcut\}}^{b}
\log\!\Big(\frac{\nu\,\SNR}{\beta}\,\lambda\Big)\,f_\beta(\lambda)\dd \lambda,
\end{align}
where the \emph{waterfilling cutoff} $\lamcut$ and the total power constraints are given by
\begin{align}
\lamcut&=\frac{\beta}{\SNR\,\nu},\notag\\
\int_{\max\{a,\lamcut\}}^{b}&\Big(\nu-\frac{\beta}{\SNR\,\lambda}\Big)f_\beta(\lambda)\dd \lambda=1.\label{eq:constraint-general}
\end{align}
The cutoff $\lamcut$ is the smallest eigenvalue that receives any power. The lower integration limit is therefore $\max\{a,\lamcut\}$.

\subsection{Solvability of the non-square cases}\label{sec:offdiag}
The non-square ratios already have a closed-form solution for two reasons.
\begin{itemize}[leftmargin=*]
    \item \textit{The soft edge is bounded away from zero}: When $\beta\neq1$, the lower support edge is strictly positive, i.e., $a=(1-\sqrt\beta)^2>0$. In random-matrix terminology, this strictly positive lower edge, where the density vanishes like a square root as $\lambda\downarrow a$, is a \emph{soft edge}. It yields a gap interval $[0,a)$ separating the smallest eigenvalues from the origin.\footnote{The contrasting
situation where the support reaches the origin and no such gap exists is a \emph{hard edge} \cite{TracyWidomLevelSpacingDistributions1994}.} Since this gap has positive width,
raising the SNR lets waterfilling reach the whole support. Since $a>0$, the $1/\lambda$ factor is bounded on $[a,b]$ and the power
constraint becomes an algebraic equation which can be solved for $\nu$ in closed form.

    \item \textit{Reciprocity simplifies the problem}: With full CSIT, the capacity is symmetric under exchange of the roles of transmitter and receiver, and we have
    \begin{align}
    \label{eq:recip}
    \C(\beta,\SNR)=\beta\,\C(1/\beta,\SNR).
    \end{align}
    The reason Eq. \eqref{eq:recip} holds is that the capacity depends on $\bH$ only through the nonzero eigenvalues of its Gram
    matrix, and the two Gram matrices $\bH^*\bH$ (of size $\Nt$) and $\bH\bH^*$ (of size $\Nr$)
    share exactly the same $\min\{\Nt,\Nr\}$ nonzero eigenvalues. Swapping the roles of transmitter and
    receiver therefore leaves the waterfilling problem over those eigenvalues unchanged. The swap does, however, map $\beta=\Nt/\Nr$ to $1/\beta$ and rescale the per-receive-antenna normalization by the factor
    $\Nt/\Nr=\beta$. Consequently, it suffices to evaluate
    the full-support integrals for $\beta<1$ while the case of $\beta>1$ branch follows by \eqref{eq:recip}
    without any new computation.
\end{itemize}

Carrying out the full-support integrals of the full-CSI capacity for $\beta<1$ and mapping to $\beta>1$ by \eqref{eq:recip} yields the known closed form \cite{GrantSomeResultsMultiAntenna2002, GrantRayleighFadingMultiAntenna2002, TulinoLozanoVerduMIMOCapacityChannel2004}

\begin{align}
\label{eq:thm35}
\C(\SNR)=
\begin{cases}
\beta\log\!\big(\tfrac{\SNR}{\beta}+\tfrac{1}{1-\beta}\big)+(1-\beta)\log\tfrac{1}{1-\beta}\\
\hspace{0.6cm}-\beta\log e, & \hspace{-0.3cm}\beta<1,\\[6pt]
\log\!\big(\beta\,\SNR+\tfrac{\beta}{\beta-1}\big)+(\beta-1)\log\tfrac{\beta}{\beta-1}\\
\hspace{0.6cm}-\log e,
& \beta>1.
\end{cases}
\end{align}

\subsection{The obstacle at the square ratio}\label{sec:square}
At the square ratio $\Nt=\Nr$, the gap interval $[0,a)$
between the origin and the smallest eigenvalue vanishes, i.e., its width $a=(1-\sqrt\beta)^2$ shrinks to zero, and at $\beta=1$ the spectrum reaches the origin.
In other words, the density blows up at the edge toward which the cutoff descends. At $\beta=1$ the Mar\v{c}enko--Pastur density \eqref{eq:mpedges} becomes
\begin{align}
\label{eq:mp1}
f_1(\lambda)=\frac{\sqrt{\lambda(4-\lambda)}}{2\pi\lambda}
=\frac{1}{2\pi}\sqrt{\frac{4-\lambda}{\lambda}},\qquad \lambda\in(0,4].
\end{align}
The $1/\lambda$ factor of the Mar\v{c}enko--Pastur law, which is of no consequence when $a>0$, now persists all the way down to the origin and makes $f_1$ unbounded as $\lambda\downarrow0$. The partial-support integral \eqref{eq:cap-general} must therefore be evaluated against a density that blows up precisely at the edge toward which the cutoff $\lamcut$ slides as $\SNR\to\infty$.

For this reason, the known result for $\beta=1$ full-CSIT capacity in closed form has been only obtained in the limit $\SNR\to\infty$ \cite{TulinoVerduRandomMatrixTheoryWireless2004}. At finite SNR, the capacity has been left computable only through a numerically-solved water level \cite{GrantSomeResultsMultiAntenna2002, JayaweeraPoorCapacityMultipleAntenna2003}, and no closed-form solution is known. The asymptotic parametric integrals appear in \cite{GrantSomeResultsMultiAntenna2002}, while the finite-antenna instance of \cite{JayaweeraPoorCapacityMultipleAntenna2003} adapts power across time under a long-term power constraint. We instead waterfill per realization under an instantaneous constraint. The two coincide in the large-antenna limit, where the empirical spectrum concentrates on its deterministic limit.

\section{Closed-form full-CSI capacity at $\beta=1$}\label{sec:angular}

One classical change of variable removes both obstacles at the square ratio at once. The obstacles are $i)$ the lower integration limit $\lamcut$, which cannot be pushed to the support edge at any finite SNR, and $ii)$ the density \eqref{eq:mp1}, which is \emph{unbounded near $\lambda=0$} at the hard edge, precisely where the cutoff lives in the high-SNR limit. We pass from the eigenvalue to the singular value, $\lambda=s^2$ with $s=2\sin\theta$, and then apply the standard trigonometric substitution for the resulting quarter-circle law \cite{MarcenkoPasturDistributionEigenvalues1967, NicaSpeicherLecturesCombinatoricsFree2006}, the same substitution that gives the Catalan-number moments of the law. Concretely,
\begin{align}
\label{eq:sub}
\lambda=4\sin^2\theta,\quad \dd \lambda=8\sin\theta\cos\theta\dd \theta,\quad
\theta\in\big(0,\tfrac\pi2\big].
\end{align}
The substitution is the natural one here. The support $\lambda \in (0,4]$ maps bijectively to $\theta \in (0,\pi/2]$, and the two edges $\lambda\in\{0,4\}$ become $\theta\in\{0,\pi/2\}$. Since $4-\lambda=4\cos^2\theta$, we
have $\sqrt{(4-\lambda)/\lambda}=\cot\theta$, and the Mar\v{c}enko--Pastur law \eqref{eq:mp1} pushes forward to a
bounded, everywhere-smooth \emph{raised-cosine} weight, i.e.:
\begin{align}
\label{eq:pushforward}
f_1(\lambda)\dd \lambda
=\frac{1}{2\pi}\cot\theta\cdot 8\sin\theta\cos\theta\dd \theta
=\frac{4}{\pi}\cos^2\theta\dd \theta .
\end{align}
The singularity at $0$ is now gone because the Jacobian $8\sin\theta\cos\theta$ vanishes at the origin at exactly the rate needed to cancel the $1/\sqrt\lambda$ blow-up, leaving a trigonometric polynomial.
The upper edge $\lambda=4$ maps to $\theta=\pi/2$, and the interior cutoff $\lamcut$ maps to a
\emph{reference angle} $\theta_0$ defined by
\begin{align}
\label{eq:theta0}
\sin^2\theta_0=\frac{\lamcut}{4}=\frac{1}{4\nu\,\SNR}
~\Longleftrightarrow~
\kappa\triangleq\nu\,\SNR=\frac{1}{4\sin^2\theta_0}.
\end{align}
The one unknown of the problem, the water level $\nu$, is now encoded in $\theta_0$ and the edge-singular integrand has become a trigonometric polynomial. In the truncated (partial-support) problem, this parametrization removes the edge singularity, represents the moving cutoff by the single angle $\theta_0$, and reduces the power-constraint integrals to trigonometric expressions. It thereby lets us eliminate $\nu$ and make the $\SNR$ explicit in $\theta_0$. Before we proceed further, we recall two integrals:
\begin{align}
J_0(\theta_0)&\triangleq\int_{\theta_0}^{\pi/2}\cos^2\theta\dd \theta
=\tfrac12\Big(\tfrac\pi2-\theta_0\Big)-\tfrac14\sin 2\theta_0,\label{eq:J0}\\
K_0(\theta_0)&\triangleq\int_{\theta_0}^{\pi/2}\cot^2\theta\dd \theta
=\cot\theta_0-\Big(\tfrac\pi2-\theta_0\Big).\label{eq:K0}
\end{align}

\subsection{SNR as a function of the reference angle $\theta_0$}\label{sec:snr}
We first turn the power constraint into an explicit formula for the $\SNR$. At $\beta=1$,
\eqref{eq:constraint-general} reads
\begin{align}
\label{eq:constr1}
\nu\int_{\lamcut}^{4}f_1(\lambda)\dd \lambda
-\frac{1}{\SNR}\int_{\lamcut}^{4}\frac{1}{\lambda}\,f_1(\lambda)\dd \lambda=1 .
\end{align}
We change both integrals to the angular variable. The first integrand is
$f_1\dd \lambda=\tfrac{4}{\pi}\cos^2\theta\dd \theta$ by \eqref{eq:pushforward}, and we get using \eqref{eq:J0}
\begin{align}
\int_{\lamcut}^{4}f_1\dd \lambda=\frac{4}{\pi}\int_{\theta_0}^{\pi/2}\cos^2\theta\dd \theta
=\frac{4}{\pi}J_0(\theta_0).
\end{align}
For the second integrand, divide \eqref{eq:pushforward} by $\lambda=4\sin^2\theta$:
\begin{align}
\frac{1}{\lambda}f_1\dd \lambda
=\frac{1}{4\sin^2\theta}\cdot\frac{4}{\pi}\cos^2\theta\dd \theta
=\frac{1}{\pi}\cot^2\theta\dd \theta,
\end{align}
and we get using \eqref{eq:K0}
\begin{align}
    \int_{\lamcut}^{4}\frac{1}{\lambda}f_1\dd \lambda=\frac{1}{\pi}K_0(\theta_0).
\end{align}
Substituting both into \eqref{eq:constr1} gives the compact relation
\begin{align}
\label{eq:constr-ang}
\frac{4}{\pi}\,\nu\,J_0(\theta_0)-\frac{1}{\pi\,\SNR}\,K_0(\theta_0)=1 .
\end{align}
From \eqref{eq:theta0}, $\nu=\kappa/\SNR$ with
$\kappa=1/(4\sin^2\theta_0)$, so the first term in \eqref{eq:constr-ang} becomes
\begin{align}
\frac{4}{\pi}\,\nu\,J_0
=\frac{4}{\pi}\cdot\frac{1}{4\sin^2\theta_0\,\SNR}\,J_0
=\frac{1}{\pi\,\SNR}\cdot\frac{J_0(\theta_0)}{\sin^2\theta_0}.
\end{align}
Multiplying \eqref{eq:constr-ang} by $\pi\,\SNR$ then isolates the $\SNR$:
\begin{align}
\label{eq:constr-ang2}
\pi\,\SNR=\frac{J_0(\theta_0)}{\sin^2\theta_0}-K_0(\theta_0).
\end{align}
Inserting \eqref{eq:J0} and \eqref{eq:K0}, and using
$\sin2\theta_0=2\sin\theta_0\cos\theta_0$ so that
$\tfrac{\sin2\theta_0}{\sin^2\theta_0}=2\cot\theta_0$ yields
\begin{align}
&\frac{J_0}{\sin^2\theta_0}-K_0\\
&\hspace{0.5cm}=\frac{\tfrac12(\tfrac\pi2-\theta_0)-\tfrac14\sin2\theta_0}{\sin^2\theta_0}
-\Big(\cot\theta_0-(\tfrac\pi2-\theta_0)\Big)\notag\\
&\hspace{0.5cm}=\Big(\tfrac\pi2-\theta_0\Big)\Big(1+\frac{1}{2\sin^2\theta_0}\Big)-\frac32\cot\theta_0 .
\label{eq:snr-simplify}
\end{align}
Dividing by $\pi$ yields the following proposition.

\begin{proposition}[Parametric SNR]
\label{prop:snr}
For each $\theta_0\in(0,\pi/2)$, the operating $\SNR$ that produces cutoff angle $\theta_0$ is
\begin{align}
\label{eq:snr-final}
\SNR(\theta_0)=\frac{1}{\pi}\left[\Big(\frac\pi2-\theta_0\Big)
\Big(1+\frac{1}{2\sin^2\theta_0}\Big)-\frac32\cot\theta_0\right].
\end{align}
The map $\theta_0\mapsto\SNR(\theta_0)$ is a decreasing bijection from $(0,\pi/2)$ onto
$(0,\infty)$ with $\SNR\to0$ as $\theta_0\to\pi/2$ and $\SNR\sim 1/(4\theta_0^2)\to\infty$ as
$\theta_0\to0$.
\end{proposition}

\begin{proof}
Differentiating \eqref{eq:snr-final} gives
\begin{align}
\pi\,\SNR'(\theta_0)
= \frac{\cos\theta_0}{\sin^3\theta_0}
  \left[\tfrac{1}{2}\sin 2\theta_0 - \left(\tfrac{\pi}{2}-\theta_0\right)\right].
\end{align}
Let $g(\theta) = \pi/2 - \theta - \tfrac{1}{2}\sin 2\theta$. Then
$g'(\theta) = -1 - \cos 2\theta \le 0$ and $g(\pi/2)=0$, so $g>0$ on $(0,\pi/2)$.
The bracket equals $-g(\theta_0)<0$, so $\SNR(\theta_0)$ is strictly decreasing.
The limits as $\theta_0 \to 0$ and $\theta_0 \to \pi/2$ give surjectivity onto $(0,\infty)$.
\end{proof}

\noindent As a result, one can parametrize the capacity curve by $\theta_0$ rather than solving the scalar waterline constraint numerically at each $\SNR$. This is exactly the step that had been unknown. The numerical relation between $\nu$ and $\SNR$ has been transformed into an explicit invertible mapping.

\subsection{Full-CSI capacity in closed form}\label{sec:cap}
We evaluate the capacity at the reference angle. Using $\kappa=\nu\SNR=1/(4\sin^2\theta_0)$ from
\eqref{eq:theta0} and $\lambda=4\sin^2\theta$, we write:
\begin{align}
\nu\,\SNR\,\lambda=\kappa\cdot4\sin^2\theta=\frac{\sin^2\theta}{\sin^2\theta_0},
\end{align}
so the water level contributes to the capacity \emph{only} through the reference angle. Substituting this
and \eqref{eq:pushforward} into \eqref{eq:cap-general} at $\beta=1$ gives
\begin{align}
\C(\theta_0)&=\frac{4}{\pi}\int_{\theta_0}^{\pi/2}
\log\!\frac{\sin^2\theta}{\sin^2\theta_0}\,\cos^2\theta\dd \theta\notag\\
&=\frac{8}{\pi}\Big[\cI_3(\theta_0)-J_0(\theta_0)\log\sin\theta_0\Big],\label{eq:Cangular}
\end{align}
where we split $\log\tfrac{\sin^2\theta}{\sin^2\theta_0}=2\log\sin\theta-2\log\sin\theta_0$ and
abbreviate
\begin{align}
\cI_3(\theta_0)=\int_{\theta_0}^{\pi/2}\log(\sin\theta)\cos^2\theta\dd \theta .
\end{align}
The only nontrivial quantity is $\cI_3$. Writing $\cos^2\theta=\tfrac{1+\cos2\theta}{2}$,
integrating by parts, and using the classical evaluation
$\int_0^{\phi}\log\sin t\dd t=-\phi\log2-\tfrac12\Cl(2\phi)$ \cite{LewinPolylogarithmsAssociated1981} in terms of the Clausen function
\begin{align}
\label{eq:clausen}
\Cl(x)=\sum_{n\ge1}\frac{\sin nx}{n^2}=\Im\,\Li(e^{\jj x}),
\end{align}
one finds the antiderivative of $\log(\sin\theta)\cos^2\theta$, which we denote $G(\theta)$,
\begin{align}
G(\theta)&=-\frac{\theta}{2}\log2-\frac14\Cl(2\theta)
+\frac{\sin2\theta}{4}\log\sin\theta\notag\\
&\quad-\frac{\theta}{4}-\frac{\sin2\theta}{8},\label{eq:antideriv}
\end{align}
so that $\cI_3(\theta_0)=G(\pi/2)-G(\theta_0)$ with $G(\pi/2)=-\tfrac\pi8(2\log2+1)$ using $\Cl(\pi)=0$ and $\sin\pi=0$. Substituting $G(\pi/2)$ and \eqref{eq:antideriv} into \eqref{eq:Cangular}, the terms proportional to $\sin2\theta_0\,\log\sin\theta_0$ cancel against the
$-J_0(\theta_0)\log\sin\theta_0$ term, and one obtains the following parametric capacity expression.

\begin{proposition}[Parametric capacity]
\label{prop:cap}
For each $\theta_0\in(0,\pi/2)$, the full-CSIT capacity (in nats per receive antenna) of the square
canonical channel at $\SNR=\SNR(\theta_0)$ is
\begin{align}
\C(\theta_0)&=-(2\log2+1)+\frac{4\theta_0}{\pi}\log2+\frac{2}{\pi}\Cl(2\theta_0)\notag\\
&~~~~+\frac{2\theta_0}{\pi}+\frac{\sin2\theta_0}{\pi}
-\Big(2-\frac{4\theta_0}{\pi}\Big)\log\sin\theta_0.\label{eq:cap-final}
\end{align}
\end{proposition}

Note how setting $\theta_0 = \pi/2 - \delta$ and expanding \eqref{eq:snr-final} and \eqref{eq:cap-final} gives
$\SNR = 2\delta^5/(15\pi) + O(\delta^7)$ and $\C = 8\delta^5/(15\pi) + O(\delta^7)$, so
$\C/\SNR \to 4$ nats as $\SNR \to 0$. This means that the low-SNR slope equals the upper spectral edge $b=4$, which is consistent with waterfilling allocating all power to the strongest eigendirections. Together with
$\C \to \log \SNR - 1$, the capacity in \eqref{eq:cap-final} is matched at both ends of the SNR range.

\noindent Together, \proprefand{prop:snr}{prop:cap} provide the closed-form parametric solution for the full-CSI capacity at $\beta=1$.

The result can be written in terms of the cutoff eigenvalue, but not in terms of the SNR. The angle is tied to the cutoff \emph{eigenvalue} by $\theta_0=\arcsin\!\sqrt{\lamcut/4}$
(equivalently $\sin\theta_0=\sqrt{\lamcut}/2$, $\cos\theta_0=\sqrt{4-\lamcut}/2$), so one may substitute this into \eqref{eq:snr-final} and \eqref{eq:cap-final} to express both the $\SNR$ and the capacity in closed form as functions of $\lamcut$, the smallest powered eigenvalue. What one \emph{cannot} do is write the capacity $\C$ as an elementary closed form directly in $\SNR$: that would require inverting the transcendental map \eqref{eq:snr-final}, $\theta_0\mapsto\SNR(\theta_0)$, which has no elementary inverse. This is exactly why the reference left the finite-SNR value only implicit, and why the natural free parameter of the closed form is the angle $\theta_0$ (or the cutoff $\lamcut$) rather than $\SNR$ itself. The pair $\big(\SNR(\theta_0),\C(\theta_0)\big)$ traces the same curve as $\C$-versus-$\SNR$ would, without ever inverting the map.

\section{Numerical validation}\label{sec:verify}
The closed form of \proprefand{prop:snr}{prop:cap} agrees with brute-force waterfilling to six decimal places over the full SNR range. The reference values come from solving the scalar constraint \eqref{eq:constraint-general} for $\nu$ at a given $\SNR$ and integrating \eqref{eq:cap-general} numerically. \figref{fig:cap} plots the capacity in nats per receive antenna versus SNR in dB. The solid blue curve is the closed form, traced by sweeping $\theta_0\in(0,\pi/2)$ and reading off the pairs $\big(\SNR(\theta_0),\C(\theta_0)\big)$. The red circles are finite-size waterfilling at $\Nt=\Nr=64$, averaged over IID complex Gaussian realizations of $\bH$. They lie on the curve across the whole range, so the asymptotic formula is already accurate at moderate antenna counts.

\begin{figure}[!h]
\centering
\vspace{-0.2cm}
\includegraphics[width=0.93\columnwidth]{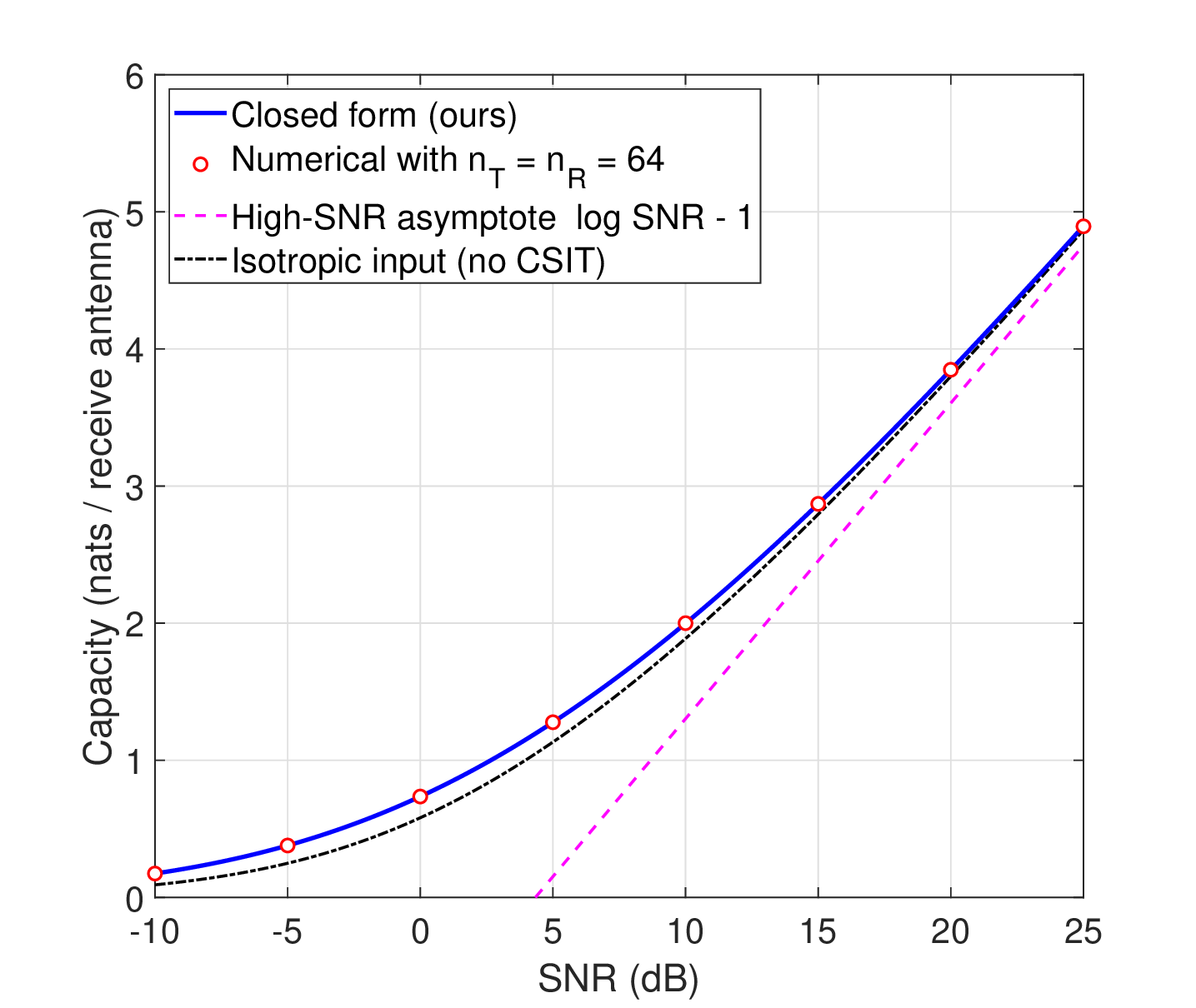}
\caption{Full-CSIT capacity of the square ($\beta=1$) canonical MIMO channel versus SNR (dB). The closed form of \proprefs{prop:snr}{prop:cap} (solid) matches brute-force finite-size waterfilling at $\Nt=\Nr=64$ (circles) across the whole SNR range. The high-SNR asymptote $\log\SNR-1$ (dashed) is tight above $\sim$20\,dB but is not meaningful at low SNR. The gap to the isotropic input without CSIT (dash-dot) is the waterfilling gain.}
\label{fig:cap}
\end{figure}

Above about $20$~dB the capacity approaches the high-SNR asymptote $\log\SNR-1$ (dashed magenta). Below that point the asymptote is not meaningful, since it goes negative at low SNR. The vertical gap between the isotropic input without CSIT (dash-dot black) and the waterfilling capacity (blue) is the waterfilling gain from full transmit CSI. That gain is widest at low SNR and closes as $\SNR\to\infty$, where the two strategies meet the asymptote.

\section{Concluding remarks}\label{sec:remarks}

In this letter we derived the full-CSI capacity of the square Mar\v{c}enko--Pastur MIMO channel in closed form, a quantity previously available only through numerical evaluation. The non-square ratios are solvable because waterfilling reduces to a full-support integral at high SNR once the interior cutoff vanishes. The square ratio is harder, since the capacity integral has no full-support representation at any finite SNR. The trigonometric parametrization itself is classical \cite{MarcenkoPasturDistributionEigenvalues1967, NicaSpeicherLecturesCombinatoricsFree2006}. What is new is the closed-form evaluation of both the SNR and the capacity at $\beta=1$, which turns the interior cutoff into an explicit parameter rather than a numerically solved water level. The capacity curve is therefore traced by sweeping one angle, with the low-SNR slope reading off the upper spectral edge.

The derivation is specific to the canonical IID channel with full CSIT at $\beta=1$. Correlated, Rician, and non-square models need their own limiting spectra, though the same change of variable applies whenever the limiting law has a hard edge at the origin, even when the resulting integrals differ. A natural next step is to link this finite-SNR description to detection-theoretic thresholds for the same square channel.

\section*{Acknowledgment}
Gemini 3.1 Pro and Opus 4.8 were used in the preparation of this manuscript for editing, grammar checking, and knowledge retrieval. No passages were copied without full author review, and all substantive ideas, analyses, and conclusions are the product and responsibility of the authors. The change of variable in \eqref{eq:sub} that led to the closed form was suggested by these tools, though the derivation and verification were performed by the authors.

\bibliographystyle{IEEEtran}
\bibliography{references.bib}    

\end{document}